\documentclass[12pt,a4paper]{amsart} 
\usepackage{amsmath,amsfonts,amssymb,amsthm,amscd,slashed}
\usepackage{epsfig}
\pdfoutput=1

\usepackage{tikz}

\usepackage{xypic}

\usepackage{hyperref}

\newtheorem{thm}{Theorem}
\newtheorem{corl}[thm]{Corollary}
 \newtheorem{prop}[thm]{Proposition}
\newtheorem{defn}[thm]{Definition}

\newtheorem{ex}[thm]{Example}
\newtheorem{rem}[thm]{Remark}

\def\1{\mathbb{I}}

\def\A{\mathcal{A}}

\DeclareMathOperator{\Ad}{Ad}

\DeclareMathOperator{\Aut}{Aut}
\def\B{\mathfrak{B}}
\def\fB{\mathfrak{B}}

\def\C{\mathbb{C}}

\def\cO{\mathcal{O}}

\def\dd{\mathrm{d}}

\DeclareMathOperator{\Der}{Der}
\def\Diff{{\textup{Diff}}}
\def\dirac{D_M}
\def\dR{{\scriptscriptstyle \mathrm{dR}}}

\def\ev{\mathrm{ev}}

\def\G{\mathfrak{G}}

\def\cG{\mathfrak{G}}
\def\g{\mathfrak{g}}

\def\H{\mathcal{H}}

\DeclareMathOperator{\Inn}{Inn}
\def\into{\hookrightarrow}

\def\nn{\nonumber}

\def\op{\textup{op}}

\DeclareMathOperator{\Out}{Out}

\def\pert{{\rm Pert}}

\def\R{\mathbb{R}}

\def\S{\mathcal{S}}

\def\bS{\mathbb{S}}

\def\su{\mathfrak{su}}

\def\bT{\mathbb{T}}

\def\tilde{\widetilde}

\def\U{\mathcal{U}}
\def\u{\mathfrak{u}}

\def\Z{\mathbb{Z}}

\title[Localizing gauge theories from NCG]{Localizing gauge theories from noncommutative geometry}
\author{Walter D. van Suijlekom}
\address{\flushleft Institute for Mathematics, Astrophysics and Particle Physics, 
Faculty of Science, Radboud University Nijmegen, Heyendaalseweg 135, 6525 AJ Nijmegen, The Netherlands}
\email{waltervs@math.ru.nl}
\date{November 24, 2014}

\begin{document}

\maketitle

\begin{abstract}
We recall the emergence of a generalized gauge theory from a noncommutative Riemannian spin manifold, {\it viz.}~a real spectral triple $(\A,\H,D;J)$. This includes a gauge group determined by the unitaries in the $*$-algebra $\A$ and gauge fields arising from a so-called perturbation semigroup which is associated to $\A$. Our main new result is the interpretation of this generalized gauge theory in terms of an upper semi-continuous $C^*$-bundle on a (Hausdorff) base space $X$. The gauge group acts by vertical automorphisms on this $C^*$-bundle and can (under some mild conditions) be identified with the space of continuous sections of a group bundle on $X$. This then allows for a geometrical description of the group of inner automorphisms of $\A$. 

We exemplify our construction by Yang--Mills theory and toric noncommutative manifolds and show that they actually give rise to continuous $C^*$-bundles. Moreover, in these examples the corresponding inner automorphism groups can be realized as spaces of sections of group bundles that we explicitly determine. 

\end{abstract}

\tableofcontents

\section{Introduction}
There is a natural link between noncommutative geometry and non-abelian gauge theories. This is mainly due to the fact that any noncommutative (involutive) algebra $\A$ gives rise to a non-abelian group of invertible (unitary) elements in $\A$. This has given rise to many applications in physics, such as to Yang--Mills theories \cite{CC96,CC97} and to the derivation of the Standard Model of elementary particles from a noncommutative Riemannian spin manifold \cite{CCM07}. 

Even though these examples deal with gauge theories on {\it commutative} background spaces, the gauge group and the gauge fields are defined along the general lines of \cite{C96} ({\it cf.} the more recent \cite{CCS13} and Section \ref{sect:st} below), valid for any noncommutative Riemannian spin manifold. For instance, the gauge group is given in terms of the unitary elements $\U(\A)$ in an involutive algebra $\A$, independent of a classical background space. However, in the physical applications of \cite{CC96,CC97,CCM07,CM07,CCS13b} ---including extensions of them to the topologically non-trivial case \cite{BS10,Cac11,BD13}--- the elements in $\U(\A)$ yield automorphisms of a principal bundle, in perfect agreement with the usual description of gauge theories. 

In the present paper we will show that a gauge theory derived from a noncommutative Riemannian spin manifold can always be described by means of bundles on a commutative background space. More precisely, starting with a so-called {\it real spectral triple} for the $C^*$-algebra $A$, we identify a subalgebra $A_J$ in the center of $A$ which by Gelfand duality is isomorphic to $C(X)$ for some compact Hausdorff topological space $X$. This turns $A$ into a so-called $C(X)$-algebra \cite{Kas88} for which it is well-known that it can be identified with the $C^*$-algebra of continuous sections of a bundle $\fB$ of $C^*$-algebras on $X$ (in general, this is an upper semi-continuous $C^*$-bundle, see references for Theorem \ref{thm:usc} below). This bundle $\fB$ will set the stage for the generalized gauge theory. We will show that the gauge group $\U(A) /\U(A_J)$ derived from the real spectral triple acts by vertical bundle automorphisms on this bundle, which agrees with the action of it on $A$ by inner automorphisms (Proposition \ref{prop:gauge}). Moreover, under some additional conditions, we identify (Theorem \ref{thm:lift}) a group bundle\footnote{In this paper we use the term {\it group bundle} for a continuous, open, surjection $\pi: \G\to X$ between topological spaces such that each fiber $\pi^{-1}(x)$ is a topological group. }  whose space of continuous sections coincides with the gauge group. The gauge fields can be considered as sections of a bundle $\fB_\Omega$ constructed in much the same way as $\fB$, also carrying an action of the gauge group which agrees with the usual gauge transformation for gauge fields (Theorems \ref{thm:one-forms} and \ref{thm:gauge-field}). 

Besides the applications to Yang--Mills theory that we recall in Section \ref{subsect:ym}, we consider the interesting class of toric noncommutative manifolds in Section \ref{subsect:toric}. They are obtained by deformation quantization of a Riemannian spin manifold $M$ along a torus action. We identify the base space of our $C^*$-bundle with the corresponding orbit space, and characterize the fiber $C^*$-algebra. We show that the $C^*$-bundle is always continuous, as opposed to merely upper semi-continuous (Theorem \ref{thm:toric-bundle}). Moreover, if the orbit space is simply connected, then the gauge group is isomorphic to the space of continuous sections of a group bundle on that orbit space (Proposition \ref{prop:gauge-theta}), which in turn is isomorphic to the group of inner automorphisms (Corollary \ref{corl:inner-theta}). We end by a concrete study of two examples: the toric noncommutative spheres $\bS^3_\theta$ and $\bS^4_\theta$.

\subsection*{Acknowledgements}
This research was partially supported by NWO under VIDI-grant 016.133.326.
I would also like to thank the Hausdorff Institute for Mathematics in Bonn for hospitality and support during the Trimester Program ``Noncommutative Geometry and its Applications'' in the Fall of 2014. I thank Simon Brain, Branimir \'Ca\'ci\'c, Alan Carey, Bram Mesland and Adam Rennie for fruitful discussions, and Jord Boeijink for discussions and a careful proofreading.

\section{Spectral triples}
\label{sect:st}
The basic device in noncommutative Riemannian spin geometry is a spectral triple (see \cite[Section IV.2.$\delta$]{C94}  where they were called unbounded $K$-cycles). Let us start by recalling its definition, including the notion of a real structure \cite{C95}.

\begin{defn}
\label{defn:st}
A {\em spectral triple} $(\A,\H,D)$ is given by a unital $*$-algebra $\A$
faithfully represented as bounded operators on a Hilbert space $\H$ and a self-adjoint
operator $D$ in $\H$ such that the resolvent $(i+D)^{-1}$ is a compact operator and $[D,a]$ is bounded for each $a\in \A$.

\medskip


A {\em real structure} 
on a spectral triple 
is an anti-linear isometry $J: \H \to \H$ such that
\begin{equation*}
J^2 = \varepsilon, \qquad JD = \varepsilon' DJ,
\end{equation*}
where the numbers $\varepsilon ,\varepsilon' 
\in \{ -1,1\}$.

Moreover, with $b^0 = J b^* J^{-1}$, we impose the {\em commutant property} and the {\em order one condition}:
\begin{align}
[a,b^0] = 0, \qquad [[D,a],b^0] = 0 ; \qquad (a,b \in \A).
\label{comm-1ord}
\end{align}
A spectral triple with a real structure is called a {\em real spectral triple}.
\end{defn}

\begin{ex}
\label{ex:can-st}
The basic example of a spectral triple is the {\bf canonical spectral triple} associated to a compact Riemannian spin manifold:
\begin{itemize}
\item $\A=C^\infty(M)$, the algebra of smooth functions on $M$;
\item $\H=L^2(S)$, the Hilbert space of square integrable sections of a spinor bundle $S \to M$;
\item $D=\dirac$, the Dirac operator associated to the Levi--Civita connection lifted to the spinor bundle.
\end{itemize}
A real structure $J$ is given by charge conjugation $J_M$.
\end{ex}

\begin{ex}
\label{ex:ym}
More generally, let $B \to M$ be a locally trivial $*$-algebra bundle on $M$ with typical fiber $M_N(\C)$, the $*$-algebra of $N \times N$ matrices with complex entries. On the Hilbert space $L^2(B \otimes S)$ one then considers the Dirac operator $D^{B}_M$ on $B \otimes S$ defined in terms of a hermitian $*$-algebra connection on $B$ and the spin connection on $S$ (cf. \cite{BoeS10} or \cite[Chapter 10]{Sui14} for more details). A real structure is given on $L^2(B \otimes S)$ by $J(s \otimes \psi) = s^* \otimes J_M \psi$ making the following set of data a real spectral triple:
$$
\left(\Gamma^\infty(B), L^2(B \otimes S), D^B_M; (\cdot)^\ast \otimes J_M \right).
$$
\end{ex}

\begin{defn}
\label{defn:unitary-equiv}
We say that two real spectral triples $(\A_1,\H_1,D_1;J_1)$ and $(\A_2,\H_2,D_2;J_2)$ are {\em unitarily equivalent} if $\A_1 = \A_2$ and if there exists a unitary operator $U: \H_1 \to \H_2$ such that
\begin{align*}
U \pi_1(a) U^*  &= \pi_2(a); \qquad ( a \in \A_1),\\
U D_1 U^* &= D_2,\\
UJ_1U^*&=J_2,
\end{align*}
where we have explicitly indicated the representations $\pi_i$ of $\A_i$ on $\H_i$ ($i=1,2$).
\end{defn}

Any spectral triple gives rise to a differential calculus on the $*$-algebra $\A$ \cite[Section VI.1]{C94} (see also \cite[Chapter 7]{Lnd97}). We focus only on differential one-forms, as this is sufficient for our applications to gauge theory below.  
\begin{defn}
\label{defn:one-forms}
The $\A$-bimodule of {\em Connes' differential one-forms} is given by
$$
\Omega^1_D(\A) := \left\{ \sum_k a_k [D,b_k]: a_k,b_k \in \A \right\},
$$
and the corresponding derivation $\dd : \A \to \Omega^1_D(\A)$ is given by $\dd = [D,\cdot]$.
\end{defn}

\begin{ex}
\label{ex:ym:1-forms}
In the case of a Riemannian spin manifold $M$ we can identify \cite{C85,C94} 
$$
\Omega^1_{\dirac}(C^\infty(M)) \cong \Omega^1_\dR(M),
$$ 
the usual De Rham differential one-forms. 

More generally, if $\A = \Gamma^\infty(M,B)$ for a locally trivial $*$-algebra bundle $B$ as in Example \ref{ex:ym}, then 
$$
\Omega^1_{D^B_M}(\A) \cong \Omega^1_\dR(M) \otimes_{C^\infty(M)} \Gamma^\infty(M,B). 
$$
That is to say, it is the space of smooth sections of the $*$-algebra bundle $B$ taking values in the one-forms on $M$.
\end{ex}

\bigskip

As a final preparation for the next sections, we recall the construction of a commutative subalgebra $\A_J$ of $\A$ \cite[Prop.~3.3]{CCM07} ({\em cf.} \cite{DS12}). 
\begin{defn}
Given a real spectral triple $(\A,\H,D;J)$ we define
$$
\A_J : = \left\{ a \in A: aJ = Ja^*\right\}.
$$
\end{defn}
As we will see shortly, this is a complex subalgebra, contained in the center of $\A$ (and hence commutative). 
\begin{rem}
The definition of the commutative subalgebra $\A_J$ is quite similar to the definition of a subalgebra of $\A$ defined in \cite[Prop.~3.3]{CCM07} ({\it cf.}\ \cite[Prop.~1.125]{CM07}), which is the {\it real} commutative subalgebra in the center of $\A$ consisting of elements for which $aJ = Ja$. Following \cite{DS12} we propose a similar but different definition, since this subalgebra will turn out to be very useful for the description of the gauge group associated to any real spectral triple, as we will describe below. 
\end{rem}
\begin{prop}
\label{prop:subalg}
Let $(\A,\H,D;J)$ be a real spectral triple. Then 
\begin{enumerate}
\item $\A_J$ defines an involutive commutative complex subalgebra of the center of $\A$. 
\item $(\A_J, \H,D;J)$ is a real spectral triple.
\item Any $a \in \A_J$ commutes with the algebra generated by the sums $\sum_j a_j [D,b_j] \in \Omega^1_D(\A)$ with $a_j,b_j \in \A$.
\end{enumerate}
\end{prop}
\begin{proof}
(1) If $a \in \A_J$, then also $a^* J = \epsilon J (Ja^*) J = \varepsilon J a J^2 = J a$, using that $J^2=\varepsilon$. Hence, $\A_J$ is involutive. Moreover, for all $a \in \A_J$ and $b \in \A$ we have $[a,b] = [Ja^*J^{-1}, b ] = 0$ by the commutant property \eqref{comm-1ord}. Thus, $\A_J$ is in the center of $\A$. 

(2) Since $\A_J$ is a subalgebra of $\A$, all conditions for a spectral triple are automatically satisfied.

(3) This follows from the order one condition in Equation \eqref{comm-1ord}: 
$$
[a, [D,b] ] = [ J a^*J^{-1}, [D,b]] = 0,
$$
for $a \in \A_J$ and $ b \in \A$.
\end{proof}

\begin{ex}
\label{ex:AJ-can-st}
In the case of a Riemannian spin manifold $M$ with real structure $J_M$ given by charge conjugation, one checks that
$$
C^\infty(M)_{J_M} = C^\infty(M). 
$$
More generally, for the spectral triples of Example \ref{ex:ym} we find that
$$
\Gamma^\infty(M,B)_{J} \cong C^\infty(M).
$$
\end{ex}

\section{Inner unitary equivalences as the gauge group}
The interpretation of the inner automorphism group as a gauge group was first presented in \cite{C96}.

\begin{defn}
An {\em automorphism} of a $*$-algebra $\A$ is a linear invertible map $\alpha: \A\to \A$ that satisfies
$$
\alpha(ab)=\alpha(a)\alpha(b), \qquad \alpha(a^*) = \alpha(a)^*. 
$$
We denote the group of automorphisms of the $*$-algebra $\A$ by $\Aut(\A)$.

An automorphism $\alpha$ is called {\em inner} if it is of the form $\alpha(a) = u a u^*$ for some element $u \in \U(\A)$ where
$$
\U(\A) = \{ u \in \A: u u ^* = u^* u =1 \}
$$
is the group of unitary elements in $\A$. The group of inner automorphisms is denoted by $\Inn(\A)$. 

The group of \emph{outer automorphisms} of $\A$ is defined by the quotient 
\begin{align*}
\Out(\A) := \Aut(\A) / \Inn(\A) . 
\end{align*}
\end{defn}
Note that $\Inn(\A)$ is indeed a normal subgroup of $\Aut(\A)$ since
\begin{align*}
\beta\circ\alpha_u\circ\beta^{-1}(a) = \beta\big(u\beta^{-1}(a)u^*\big) = \beta(u)a\beta(u)^* = \alpha_{\beta(u)} (a),
\end{align*}
for any $\beta \in \Aut(\A)$. 

An inner automorphism $\alpha_u$ is completely determined by the unitary element $u\in \U(\A)$, but not in a unique manner. More precisely, the map $\phi\colon \U(\A) \rightarrow \Inn(\A)$ given by $u \mapsto \alpha_u$ is surjective, but has a kernel. In fact, $\ker \phi$ is given by $\U(Z(\A))$ where $Z(\A)$ is the center of $\A$. We have thus proven the following 

\begin{prop}
There is the following isomorphism of groups:
\begin{align}
\label{eq:Inn_A}
\Inn(\A) \cong \U(\A) / \U(Z(\A)) .
\end{align} 
\end{prop}

\begin{ex}
\label{ex:auto-diffeo}
If $\A$ is a commutative $*$-algebra, then there are no non-trivial inner automorphisms since $Z(\A) = \A$. Moreover, if $\A=C^\infty(X)$ with $X$ a smooth compact manifold, then $\Aut(\A) \cong \Diff(X)$, the group of diffeomorphisms of $X$. Explicitly, a diffeomorphism $\phi :X \to X$ yields an automorphism by pullback of a function $f$:
$$
\phi^*(f)(x) = f(\phi(x)) ; \qquad (x \in X). 
$$
For a precise proof of the isomorphism between $\Aut(C(X))$ and the group of homeomorphisms of $X$, we refer to \cite[Theorem II.2.2.6]{Bla06}. For a more detailed treatment of the smooth analogue, we refer to \cite[Section 1.3]{GVF01}.
\end{ex}

\begin{ex}
\label{ex:auto-matrix}
At the other extreme, we consider an example where all automorphisms are inner. Let $\A= M_N(\C)$ and let $u$ be an element in the unitary group $U(N)$. Then $u$ acts as an automorphism on $a \in M_N(\C)$ by sending $a \mapsto u a u^*$. If $u = \lambda \1_N$ is a multiple of the identity with $\lambda \in U(1)$, this action is trivial. In fact, the group of automorphisms of $\A$ is the projective unitary group $PU(N)= U(N)/U(1)$, in concordance with \eqref{eq:Inn_A}. 
\end{ex}

Elements in $\U(\A)$ not only act on the $*$-algebra $\A$ as inner automorphisms, via the representation $\pi$ of $\A$ on $\H$ they also act on the Hilbert space $\H$ that is present in the spectral triple. In fact, with $U = \pi(u) J \pi(u) J^{-1}$, the unitary $u$ induces a unitary equivalence of real spectral triples in the sense of Definition \ref{defn:unitary-equiv}. More specifically, for such a $U$ we have 
\begin{align}
\label{eq:conj-a-U}
U \pi(a)U^* &= \pi \circ \alpha_u(a),\\
UJU^* &= J \nn.
\end{align}
We conclude that an inner automorphism $\alpha_u$ of $\A$ induces a unitarily equivalent spectral triple $(\A, \H, UDU^*; J)$, where the action of the $*$-algebra is given by $\pi\circ\alpha_u$. Note that the grading and the real structure are left unchanged under these {\em inner} unitary equivalences; only the operator $D$ is affected by the unitary transformation. For the latter, we compute, using \eqref{comm-1ord},
\begin{equation}
\label{eq:D-pure-gauge}
D \mapsto UDU^* = D + u [D,u^*] + \epsilon' J u [D,u^*] J^{-1},
\end{equation}
where as before we have suppressed the representation $\pi$. 
We recognize the extra terms as {\it pure gauge} fields $u \dd u^*$ in the space of Connes' differential one-forms $\Omega^1_D(\A)$ of Definition \ref{defn:one-forms}. This motivates the following definition 

\begin{defn}
\label{defn:gauge-ncg}
The {\em gauge group} $\cG(\A,\H;J)$ of the spectral triple is 
\begin{align*}
\cG(\A,\H;J) := \left\{ U=uJuJ^{-1} \mid u\in \U(\A) \right\} .
\end{align*}
\end{defn}
\begin{prop}
\label{prop:ses-gauge}
There is a short exact sequence of groups
$$
1 \to \U(\A_J) \to \U(\A) \to \cG(\A,\H;J) \to 1.
$$
Moreover, there is a surjective map $\cG(\A,\H;J) \to \Inn(\A)$.
\end{prop}
\proof
Consider the map $\Ad\colon \U(\A)\rightarrow \cG(\A,\H;J)$ given by sending $u\mapsto uJuJ^{-1}$. This map $\Ad$ is a group homomorphism, since the commutation relation $[u,JvJ^{-1}]=0$ of \eqref{comm-1ord} implies that 
$$
\Ad(v)\Ad(u) = vJvJ^{-1}uJuJ^{-1} = vuJvuJ^{-1} = \Ad(vu).
$$ 
By definition $\Ad$ is surjective, and $\ker(\Ad) = \{u\in \U(\A) \mid uJuJ^{-1}=1\}$. The relation $uJuJ^{-1}=1$ is equivalent to $uJ=Ju^*$ which is the defining relation of the commutative subalgebra $\A_J$. This proves that $\ker(\Ad) = \U(\A_J)$. 
For the second statement, note that the map $\cG(\A,\H;J) \to \Inn(\A)$ is given by Equation \eqref{eq:conj-a-U}, from which surjectivity readily follows. 
\endproof

\begin{corl}
\label{corl:gauge-ncg}
If $\A_J = Z(\A)$, then $\cG(\A,\H;J) \cong \Inn(\A)$. 
\end{corl}
\proof
This is immediate from the above Proposition and Equation \eqref{eq:Inn_A}.
\endproof

\begin{ex}
\label{ex:ym:gauge-group}
For the algebra $\A=\Gamma^\infty(M,B)$ appearing in Example \ref{ex:ym} we have 
$$
\A_J = Z(\A).
$$
Hence in this case the gauge group $\cG(\A,\H;J)$ coincides with the inner automorphisms of $\A$. Moreover, it turns out that $\cG(\A,\H;J)$ is isomorphic to the space of smooth sections of a $PU(N)$-bundle on $M$ with the same transition functions as $B$, at least when $M$ is simply connected ({\it cf.} \cite{BD13}). Moreover, one can reconstruct a principal $PU(N)$-bundle $P$ for which $B$ is an associated bundle:
$$
B \cong P \times_{PU(N)} M_N(\C)
$$
For more details, we refer to \cite{BoeS10} or \cite[Chapter 10]{Sui14}. 
\end{ex}



\subsection{The gauge algebra}
A completely analogous discussion applies to the definition of a gauge Lie algebra, where instead of automorphisms we now take {\bf derivations} of $\A$. The following definition essentially gives the infinitesimal version of $\cG(\A,\H;J)$.
\begin{defn}
\label{defn:gauge-alg-ncg}
The {\em gauge Lie algebra} $\g(\A,\H;J)$ of the spectral triple is
\begin{align*}
\g(\A,\H;J) := \left\{ T=X +J XJ^{-1} \mid X\in \u(\A) \right\},
\end{align*}
where $\u(\A)$ consists of the skew-hermitian elements in $\A$. 
\end{defn}
One easily checks using the commutant property,
$$
[T,T'] = [X,X'] + J [X,X'] J^{-1},
$$
so that $\g(\A,\H;J)$ is indeed a Lie algebra.

\begin{prop}
\label{prop:ses-gauge-alg}
There is a short exact sequence of Lie algebras
$$
0 \to \u(\A_J) \to \u(\A) \to \g(\A,\H;J) \to 0.
$$
\end{prop}
There are also inner derivations of $\A$ that are of the form $a \to [X,a]$; these form a Lie subalgebra $\Der_{\Inn}(\A)$ of the Lie algebra of all derivations $\Der(\A)$. If $\A_J = Z(\A)$ then 
$$
\g(\A,\H;J) \cong \Der_{\Inn}(\A),
$$
which is essentially the infinitesimal version of Corollary \ref{corl:gauge-ncg}.

 \section{Inner perturbations as gauge fields}
 \label{sect:pert}

We have seen that a non-abelian gauge group appears naturally when the $*$-algebra $\A$ in a real spectral triple is noncommutative. This gauge group acts naturally on the operator $D$ in the real spectral triple, giving rise to pure gauge field $u \dd u^*$ as perturbation of $D$. We now extend this action to obtain more general gauge fields, initially derived in \cite{C96}. A key role is played by the perturbation semigroup defined in \cite{CCS13} that is associated to the $*$-algebra $\A$ and extends the unitary group of $\A$. First, we need the following 
\begin{defn}
We denote by $\A^\op$ the opposite algebra: it is identical to $\A$ as a vector space, but with opposite product: 
$$
a^\op b^\op = (ba)^\op.
$$
Here $a^\op,b^\op$ are the elements in $\A^\op$ corresponding to $a,b \in \A$, respectively. 
\end{defn}

\begin{defn}
Let $\A$ be a unital $*$-algebra. Then we define the {\em perturbation semigroup} as
$$
\pert(\A) := \left\{ \sum_j a_j \otimes b_j^{\mathrm{op}}  \in \A \otimes \A^\op  \left| \begin{array}{l} \sum_j a_j b_j = 1  \\ \sum_j  a_j \otimes b_j^{\mathrm{op}} = \sum_j  b_j^* \otimes a_j^{*\mathrm{op}} \end{array}\right. \right\}
$$
with semigroup law coming from the product in $\A \otimes \A^\op$. 
\end{defn}

Given a spectral triple, the perturbation semigroup $\pert(\A)$ generates perturbations of the operator $D$:
$$
D \mapsto \sum_j a_j D b_j  = D + \sum_j a_j [D,b_j]
$$
For real spectral triples we use the map $\pert(\A) \to \pert(\A \otimes J \A J^{-1})$ sending $T \mapsto T \otimes JT J^{-1}$ so that
$$
D \mapsto \sum_{i,j} a_i \hat{a}_j D b_i \hat {b}_j
$$ 
where $\hat{a}_j = J a_j J^{-1}$ and $\hat{b}_j = J b_j J^{-1}$. One can show using the first-order condition that this reduces to 
$$
D_\omega := D+\omega + \epsilon' J\omega J^{-1},
$$
where $\omega^* =\omega := \sum_j a_j[D,b_j] \in \Omega^1_D(\A)$ is called the {\bf gauge field}. Alternatively, $\omega$ is called an {\bf inner perturbation} of the operator $D$, since it is the algebra $\A$ that generates the field $\omega$.

\begin{prop}
A unitary equivalence of $(\A,\H,D;J)$ as implemented by $U = u JuJ^{-1}$ with $u \in \U(\A)$ is a special case of a perturbation of $D$ by an element in $\pert(\A)$, namely, by $u \otimes u^{*\op} \in \pert(\A)$. 
\end{prop}
\proof
This follows upon inserting $\omega=u[D,u^*]$ in the above formula for $D_\omega$, yielding \eqref{eq:D-pure-gauge}.
\endproof

In the same way there is an action of the unitary group $\U(\A)$ on the new spectral triple $(\A, \H, D_\omega$) by unitary equivalences. Recall that $U= u J u J^{-1} $ acts on $D_\omega$ by conjugation:
\begin{equation*}
\label{eq:gauge-transf}
D_\omega \mapsto U D_\omega U^*.
\end{equation*}
This is equivalent to 
$$
\omega \mapsto u \omega u^* + u [D,u^*],
$$
which is the usual rule for a gauge transformation on a gauge field.
 
\begin{ex}
\label{ex:ym:pert}
We analyze the inner perturbations of $D^B_M$ for the real spectral triple introduced in Example \ref{ex:ym}. By Example \ref{ex:ym:1-forms} we have that $\omega$ can be identified with a self-adjoint one-form valued section of $B$, {\it i.e.} 
$$
\omega^* = \omega \in \Omega^1_\dR(M) \otimes_{C^\infty(M)} \Gamma^\infty(M,B). 
$$
The combination $\omega + \epsilon' J \omega J^{-1}$ amounts to the action of $\omega$ in the adjoint representation, with conjugation by $J$ implementing the right action of $\omega$ on $L^2(B \otimes S)$. This ensures that only the $\su(N)$-part of $\omega$ is relevant as an inner perturbation to $D^B_M$. With Example \ref{ex:ym:gauge-group} we can conclude that the gauge fields are given by one-form valued sections of the associated bundle $P \times_{PU(N)} \su(N)$. This brings us back to the usual definition of a gauge field in the context of principal bundles, see for instance \cite{Ble81} or \cite{AB83}. Again, full details can be found in \cite{BoeS10}. 
\end{ex}

\section{Localization}
\label{sect:localization}

Recall from Section \ref{sect:st} the construction of a complex subalgebra $\A_J$ in the center of $\A$ from a real spectral triple $(\A,\H,D;J)$. 
As $\A_J$ is commutative, Gelfand duality ({\it cf.} \cite{Bla06} or \cite{Tak02}) ensures the existence of a compact Hausdorff space such that $\A_J \subset C(X)$ as a dense $*$-subalgebra. Indeed, the norm completion of $\A_J$ is a commutative $C^*$-algebra and hence isomorphic to such a $C(X)$. We consider this space $X$ to be the `background space' on which $(\A,\H,D;J)$ describes a gauge theory.

Heuristically speaking, the above gauge group $\cG(\A,\H;J)$ considers only transformations that are `vertical', or `purely noncommutative' with respect to $X$, quotienting out the unitary transformations of the commutative subalgebra $\A_J$. In this section we make this precise by identifying a bundle $\fB \to X$ of $C^*$-algebras such that:
\begin{itemize}
\item the space of continuous sections $\Gamma(X,\fB)$ forms a $C^*$-algebra isomorphic to $A = \overline \A$, the $C^*$-completion of $\A$;
\item the gauge group acts as bundle automorphisms covering the identity. 
\end{itemize}
Moreover, we define a group bundle whose space of continuous sections is (under some conditions) isomorphic to the gauge group, and a bundle of $C^*$-algebras of which the gauge fields $\omega \in \Omega^1_D(\A)$ are sections on which the gauge group again acts by bundle automorphisms. 

We avoid technical complications that might arise from working with dense subalgebras of $C^*$-algebras, and work with the $C^*$-algebras $A_J$ and $A$ themselves, as completions of $\A_J$ and $\A$, respectively.

\medskip 

First, note that there is an inclusion map $C(X) \cong A_J \into A$. This means that $A$ is a so-called $C(X)$-algebra \cite{Kas88} which (in the unital case) is by definition a $C^*$-algebra $A$ with a map from $C(X)$ to the center of $A$. Indeed, it follows from Proposition \ref{prop:subalg} that $A_J$ is contained in the center of $A$. In such a case, it is well known that $A$ is the $C^*$-algebra of continuous sections of an upper semi-continuous $C^*$-bundle over $X$. We will briefly sketch the setup, referring to {\it e.g.}~the Appendix C in \cite{Wil07} for more details. 
Recall that a function $f:Y \to \R$ on a topological space $Y$ is {\bf upper semi-continuous} if $\{ y \in Y : f(y) <r \}$ is open for all $r \in \R$. 


\begin{defn}
\label{defn:C*-bundle}
An {\em upper semi-continuous $C^*$-bundle} over a compact topological space $X$ is a continuous, open, surjection $\pi: \fB \to X$ together with operations and norms that turn each fiber $\fB_x = \pi^{-1}(x)$ into a $C^*$-algebra, such that the map $a \mapsto \| a \|$ is upper semi-continuous and all algebraic operations are continuous on $\fB$.

A (continuous) {\em section} of $\fB$ is a (continuous) map $s: X \to \fB$ such that $\pi(s(x)) = x$. The vector space of continuous sections of $\fB$ is denoted by $\Gamma(X,\fB)$.
\end{defn}

A base for the topology on $\fB$ is given by the following collection of open sets:
\begin{equation}
\label{eq:base-top-fB}
W(s, \cO, \epsilon) := \{ b \in \fB: \pi(b) \in \cO \text{ and } \| b - s(\pi(b)) \| < \epsilon \},
\end{equation}
indexed by continuous sections $s\in \Gamma( X, \fB)$, open subsets $\cO \subset X$ and $\epsilon >0$. This heavily relies on the property that $\fB$ has {\em enough continuous sections}: for each element $b \in \fB$ there exists $s \in \Gamma(X,\fB)$ such that $s(\pi(b)) = b$. This property was established for upper semi-continuous $C^*$-bundles in \cite[Proposition 3.4]{Hof77}.

\begin{prop}
\label{prop:C*-bundle}
The space $\Gamma(X,\fB)$ of continuous sections forms a $C^*$-algebra when it is equipped with the norm
$$
\| s \| := \sup_{x \in X} \| s(x) \|_{\fB_x}.
$$
\end{prop}
\proof
See \cite{KW95,Nil96} (see also Appendix C in \cite{Wil07}) for a proof of this result.
\endproof
In our case, after identifying $A_J$ with $C(X)$, we can define a closed two-sided ideal in $A$ by
\begin{equation}
\label{eq:ideal-Ix}
I_x :=  \left \{ f a : a \in A, f \in C(X), f(x) = 0 \right \}^{-}.
\end{equation}
We think of the quotient $C^*$-algebra $\fB_x := A/ I_x$ as the fiber of $A$ over $x$ and set 
\begin{equation}
\label{eq:C*-fib}
\fB := \coprod_{x \in X} \fB_x,
\end{equation}
with an obvious surjective map $\pi: \fB \to X$. If $a \in A$, then we write $a(x)$ for the image $a+I_x$ of $a$ in $\fB_x$, and we think of $a$ as a section of $\fB$. The fact that all these sections are continuous and that elements in $A$ can be obtained in this way is guaranteed by the following result.

\begin{thm}
\label{thm:usc}
Let $(\A,\H,D;J)$ be a real spectral triple and let $A_J \cong C(X)$. The above map $\pi: \fB \to X$ with $\fB$ as in Equation \eqref{eq:C*-fib} defines an upper semi-continuous $C^*$-bundle over $X$. Moreover, there is a $C(X)$-linear isomorphism of $A$ onto $\Gamma(X, \fB)$.
\end{thm}
\proof
This follows from a general result valid for any $C(X)$-algebra $A$, realizing it as the space of continuous sections of an upper semi-continuous $C^*$-bundle on $X$. We refer to \cite{KW95,Nil96} (see also Appendix C in \cite{Wil07}) for its proof. 
\endproof

Having obtained the $C^*$-algebra $A$ as the space of sections of a $C^*$-bundle, we are ready to analyze the action of the gauge group on $A$. 
\begin{defn}
The {\bf continuous gauge group} is defined by
$$
\cG(A,\H;J) \cong \frac{\U(A)}{\U(A_J)}.
$$
\end{defn}
This contains the gauge group $\cG(\A,\H;J)$ of Definition \ref{defn:gauge-ncg} as a dense subgroup in the topology induced by the $C^*$-norm on $A$. The next result realizes the gauge group as a group of vertical bundle automorphisms of $\fB$. 

\begin{prop}
\label{prop:gauge}
The action $\alpha$ of $\cG(A,\H;J)$ on $A$ by inner automorphisms induces an action $\tilde\alpha$ of $\cG(A,\H;J)$ on $\fB$ by continuous bundle automorphisms that cover the identity. In other words, for $g \in \cG(A,\H;J)$ we have
$$
\pi(\tilde \alpha_g(b)) = \pi(b); \qquad (b \in \fB).
$$
Moreover, under the identification of Theorem \ref{thm:usc} the induced action $\tilde\alpha_*$ on $\Gamma(X, \fB)$ given by 
$$
\tilde \alpha_{g*} (s)(x) = \tilde\alpha_g (s(x))
$$
coincides with the action $\alpha$ on $A$.
\end{prop}
\proof
The action $\alpha$ induces an action on $A/I_x = \pi^{-1}(x)$, since $\alpha_g (I_x) \subset I_x$ for all $g \in \cG(A,\H;J)$. We denote the corresponding action of $\cG(A,\H;J)$ on $\fB$ by $\tilde \alpha$, so that, indeed,
$$
\pi( \tilde \alpha_g (b)) = \pi (b) ; \qquad (b \in \pi^{-1}(x)). 
$$
Let us also check continuity of this action. In terms of the base $W(s, \cO, \epsilon)$ of \eqref{eq:base-top-fB}, we find that
$$
\tilde \alpha_g (W(s, \cO, \epsilon)) = W(\tilde \alpha_{g*} (s), \cO, \epsilon), 
$$
mapping open subsets one-to-one and onto open subsets. 

For the second claim, it is enough to check that the action $\tilde \alpha_*$ on the section $s: x \mapsto  a+ I_x \in \fB_x$, defined by an element $a \in A$, corresponds to the action $\alpha$ on that $a$. In fact, 
$$
\tilde \alpha_{g*} (s)(x) = \tilde \alpha_g (s(x)) = \alpha_g (a+I_x) = \alpha_g(a) + I_x,
$$
which completes the proof.
\endproof

The derivations in the gauge algebra $\g(\A,\H;J)$ also act vertically on the upper semi-continuous $C^*$-bundle $\fB$ defined in \eqref{eq:C*-fib}, and the induced action on the sections $\Gamma(X,\fB)$ agrees with the action of $\g(\A,\H;J)$ on $A$. 
\subsection{Localization of the gauge group}

We now investigate whether or when $\G(\A,\H;J)$ can be considered as the group of continuous sections of a group bundle on the same base space $X$. Set-theoretically, one expects the group bundle that corresponds to $\Gamma(X,\fB)$ to be given by
$$
\G\B:=\coprod_{x\in X} \frac{\U(\fB_x)}{\U(\C)}.
$$
We define a topology on $\G\B$ as follows. First, the group bundle
$$
\U\B := \coprod_{x\in X}\U(\fB_x)
$$
is equipped with the induced topology from $\B$. 
Since each $\fB_x$ is a complex unital algebra, we have $\U(\C) \subset \U(\fB_x)$ so that we have a group subbundle $\coprod_{x\in X}\U(\C) \subset \U\fB$. We also write $\U\C$ for this group subbundle. The topology of $\G\B$ is then the quotient topology of the bundle $\U\B$ by the fiberwise action of the group bundle $\U\C$.

Before stating our main result on the structure of the gauge group, we consider the spaces of continuous sections of the group bundles $\U\C$ and $\U\B$. 

\begin{prop}
\label{prop:group-bundles}
We have the following group isomorphisms:
\begin{align*}
\Gamma(X,\U\C) &\cong \U(A_J),\\
\Gamma(X,\U\B) &\cong \U(A).
\end{align*}
\end{prop}
\proof
Firstly, a continuous map from $X$ to $\U(\C)$ is simply given by a unitary continuous function on $X$. Secondly, since $\Gamma(X,\B) \cong A$, unitarity translates from the product in $A$ to the fiberwise product in $\B$, hence proving the result. 
\endproof

We also need the following well-known result on covering spaces ({\it cf.} \cite[Proposition 1.33]{Hat02}).
\begin{prop}
\label{prop:covering}
Suppose given a covering space $p: (\tilde Y,\tilde y_0) \to (Y,y_0)$ and a map $f:(X,x_0) \to (Y,y_0)$ with $X$ path-connected and locally path-connected. Then a lift $\tilde f:(X,x_0) \to (\tilde Y,\tilde y_0)$ of $f$ exists if and only if $f_\ast (\pi_1(X,x_0)) \subset p_\ast (\pi_1(\tilde Y,\tilde y_0))$. 
\end{prop}

The following result generalizes a result of \cite{BD13} on Lie group bundles to the general setting of group bundles. 
\begin{thm}
\label{thm:lift}
If $X$ is simply connected and if there exists a subbundle $\tilde{\cG\B} \subset \U\B$ that is a covering space of $\cG\B$ (via the quotient map $\U\B \to \G\B$), then there is the following short exact sequence of groups
\begin{equation}
\label{eq:ses}
\xymatrix{1 \ar[r] & \Gamma(X,\U\C) \ar[r] &\Gamma(X,\U\B) \ar[r] &\Gamma(X, \G\B) \ar[r] &1.}
\end{equation}
Consequently, in this case the gauge group is given as the space of continuous sections of the group bundle $\G\B$, {\it i.e.}
$$
\G(A,\H;J) \cong \Gamma(X,\G\B).
$$
\end{thm}
\proof
Exactness of \eqref{eq:ses} is clear from the very definition of the group bundle $\G\B$, except perhaps for the claim of surjectivity of the map $\Gamma(X,\U\B) \to \Gamma(X,\G\B)$. This follows from Proposition \ref{prop:covering}, applied to a continuous section $g \in \Gamma(X,\G\B)$. Indeed, since $\pi_1( X)$ is trivial, there always exists a lift $\tilde g : X \to \tilde {\cG\B} \subset \U\B$, thus proving surjectivity. 

For the second statement, exactness of the sequence implies that
$$
\Gamma(X,\G\B) \cong \frac{\Gamma(X,\U\B)}{\Gamma(X,\U\C)} \cong \frac{\U(A)}{\U(A_J)} 
$$
using Proposition \ref{prop:group-bundles}. But this is precisely the definition of the group $\G(A,\H;J)$. 
\endproof

This result allows for the following refinement of Proposition \ref{prop:gauge}.
\begin{corl}
\label{corl:gauge}
Under the same conditions as in Theorem \ref{thm:lift}, the action of the gauge group $\G(A,\H;J)$ on $A$ is induced by the action of the fibers $\G\B_x := \U(\B_x)/\U(\C)$ on the fibers $\fB_x$ of $\fB$ by inner automorphisms. 
\end{corl}
\proof
Let $g \in \G(A,\H;J)$ with pre-image $u \in\U(A)$, {\it i.e.} so that $\alpha_g(a) = u a u^*$. Then $g, u$ and $a$ can be considered as continuous sections of bundles $\G\B, \U\B$ and $\B$ on $X$, respectively. At a point $x \in X$ we have
$g(x) \in \G\B_x = \U(\B_x)/\U(\C)$ with pre-image $u(x)\in \U(\B_x)$ and we compute as sections of $\B \to X$:
$$
\left(\alpha_g(a)\right)(x) =  u(x) a(x) u(x)^*,
$$ 
thus establishing the result.
\endproof

Note that Theorem \ref{thm:usc} also gives a bundle description of $\Inn(A)$ if $Z(A)=A_J$. Indeed, in combination with Corollary \ref{corl:gauge} we find that then $\Inn(A) \cong\Gamma(X,\G\fB)$, realizing the group of inner automorphisms of $A$ as the space of continuous sections of a group bundle.


\subsection{Localization of gauge fields}
Also the gauge fields $\omega$ that enter as inner perturbations of $D$ can be parametrized by sections of some bundle of $C^*$-algebras. In order for this to be compatible with the vertical action of the gauge group found above, we will write any gauge field as $\omega_0 + \omega$ where $\omega_0, \omega \in \Omega^1_D(\A)$ and we call $\omega_0$ the background gauge field. The action of a gauge transformation on $\omega_0+\omega$ then induces the following transformation:
$$
\omega_0 \mapsto u \omega_0 u^* + u [D,u^*]; \qquad \omega \mapsto u \omega u^*.
$$

\begin{defn}
Let $(\A,\H,D)$ be a spectral triple. We denote by $C_D(\A)$ the $C^*$-algebra generated by $\A$ and $[D,\A]$. It is a $\Z_2$-graded $C^*$-algebra by letting $a \in \A$ have degree $0$ and $[D,a]$ have degree $1$.
\end{defn}

This notion was used in \cite{LRV12} as a generalization of the Clifford algebra. Indeed, for the canonical spectral triple $C_{D_M}(C^\infty(M))$ coincides with the Clifford algebra.

\begin{thm}
\label{thm:one-forms}
Let $(\A,\H,D;J)$ be a real spectral triple with $A_J \cong C(X)$. Then the following hold:
\begin{enumerate}
\item The $C^*$-algebra $C_D(\A)$ is a (graded) $C(X)$-algebra. 
\item There is a upper semi-continuous $C^*$-bundle $\fB_\Omega$ over $X$, explicitly given by
$$ 
\fB_\Omega = \prod_{x \in X} C_D(\A)/I_x',
$$ 
where $I_x'$ is the two-sided ideal in $C_D(\A)$ generated by $I_x$ defined in Equation \eqref{eq:ideal-Ix}. 
\item Every fiber $(\fB_\Omega)_x$ is a $\Z_2$-graded $C^*$-algebra with the grading induced by the grading on $C_D(\A)$. 
\item  The $\Z_2$-graded $C^*$-algebra $\Gamma(X,\fB_\Omega)$ of continuous sections is isomorphic to $C_D(\A)$. 
\end{enumerate}
Consequently, there is a subbundle $\fB_{\Omega^1}\subset \fB_{\Omega}$ defined as the closed span of sections given by elements $\omega \in \Omega^1_D(\A) \subset C_D(\A)$. 
\end{thm}
\proof
The fact that $C_D(\A)$ is a $C(X)$-algebra follows from Proposition \ref{prop:subalg}, stating that $A_J$ commutes with both $\A$ and $[D,\A]$. The construction of a $C^*$-bundle and the claimed isomorphism then follow as in Theorem \ref{thm:usc} above. The fact that this isomorphism respects the $\Z_2$-grading follows when one considers an element $\omega \in C_D(\A)$ as a section of $\fB_\Omega$ as follows:
$$
\omega(x) = \omega + I_x',
$$ 
noting that the degree of $\omega(x)$ is induced by degree of $\omega$. 
\endproof

This allows for the following geometrical description of the gauge fields $\omega$ that arise as inner perturbations of $D$.

\begin{thm}
\label{thm:gauge-field}
Let $\pi: \fB_\Omega \to X$ be as above and let $\omega \in \Omega^1(\A)$ be understood as a continuous section of $\fB_{\Omega^1}$. Then the gauge group $\cG(A,\H;J)$ acts fiberwise on $\fB_{\Omega^1}$ and we have 
$$
\omega(x) \mapsto (u \omega u^* )(x) = u \omega(x) u^*; \qquad (x \in X),
$$
for an element $uJuJ^{-1} \in \cG(\A,\H;J)$. 
\end{thm} 

\begin{corl}
If the conditions of Theorem \ref{thm:lift} are satisfied, then the action of the gauge group $\cG(A,\H;J)$ on $\omega \in \Omega^1_D(\A)$ is induced from the action of the fibers $\G\B_x = \U(\B_x) / \U(\C)$ on the fibers $(\fB_\Omega)_x$ by
$$
\omega(x) \mapsto u(x) \omega(x) u(x)^*.
$$
\end{corl}
\proof
This follows by complete analogy with Corollary \ref{corl:gauge} above.
\endproof

\section{Applications}

\subsection{Yang--Mills theory}
\label{subsect:ym}
Consider the real spectral triple of Example \ref{ex:ym}:
$$
(\Gamma^\infty(M,B), L^2(M,B \otimes S), D_M^B;  (\cdot)^* \otimes J)
$$
with $B \to M$ a locally trivial $*$-algebra bundle with typical fiber $M_N(\C)$ and $D_M^B$ is the Dirac operator on $M$ with coefficients in $B$. 

Let us summarize what we have already established in Examples \ref{ex:ym:gauge-group} and \ref{ex:ym:pert}:
\begin{itemize}
\item There is a $PU(N)$-principal bundle $P$ such that we have the following isomorphism of locally trivial $*$-algebra bundles:
$$
B \cong P \times_{PU(N)} M_N(\C).
$$ 
\item If $M$ is simply connected, then the gauge group $\G(\A,\H;J)$ is isomorphic to sections of the adjoint bundle to this principal bundle, that is to say, 
$$
\G(\A,\H;J) \cong \Gamma(M,P \times_{PU(N)} PU(N)).
$$
so that the group bundle $\G B \cong P \times_{PU(N)} PU(N)$.
\item The inner perturbations of $D_M^B$ are parametrized by sections of the associated bundle $P \times_{PU(N)} \su(N)$, with values in $\Omega^1(M)$. Moreover, we have an isomorphism
$$
(P \times_{PU(N)} \su(N)) \otimes \Lambda^1(M) \cong B_{\Omega^1},
$$ 
where $B_\Omega$ is the tensor product of $B$ with the Clifford bundle on $M$. 
Moreover, the action of $\G(\A,\H;J)$ on $B_{\Omega^1}$ agrees with the usual gauge action $\omega \mapsto u \omega u^*$ induced by the action of $PU(N)$ on $\su(N)$.
\end{itemize}

In conclusion, for this example our generalized gauge theory obtained in Theorems \ref{thm:usc}, \ref{thm:one-forms} and \ref{thm:gauge-field} agrees with the usual principal bundle description of gauge theories ({\it cf.} \cite{Ble81}). This example was first introduced in \cite{BoeS10}, following the globally trivial example of \cite{CC96,CC97}. A more general class of examples ---so-called almost-commutative manifolds--- is studied in \cite{Cac11,BD13}.

\subsection{Toric noncommutative manifolds}
\label{subsect:toric}
A less trivial example is given by the noncommutative manifolds introduced by Connes and Landi \cite{CL01}, further studied in \cite{CD02}. These were based on deformations of $C^*$-algebras by actions of a torus, as studied by Rieffel in \cite{Rie93a}. Starting point is the {\em noncommutative torus}, which dates back already to \cite{C80,Rie81}.

\begin{defn}
The $C^*$-algebra $A_\theta$ is defined to be the $C^*$-algebra generated by two unitaries $U_1$ and $U_2$ with the defining relation
$$
U_2 U_1 = e^{2 \pi i \theta} U_1 U_2,
$$
where $\theta$ is a real parameter.
\end{defn}
For irrational values of $\theta$, the $C^*$-algebra $A_\theta$ is called the {\em irrational rotation algebra}. More generally, for any $\theta \neq 0$ it is referred to as the {\em noncommutative 2-torus} and we also write $C(\bT^2_\theta)$ instead of $A_\theta$. This is to illustrate the fact that if $\theta=0$ we obtain the commutative 2-torus $\mathbb T^2$ and we consider $A_\theta$ as a deformation of $C(\mathbb T^2)$. One can also consider such deformations at the smooth level, and define a Fr\'echet $*$-algebra by
$$
\A_\theta = \left\{ \sum_{n_1 ,n_2 \in \Z} a_{n_1n_2} U_1^{n_1}U_2^{n_2}: (a_{n_1 n_2}) \in \S(\Z^2)\right\},
$$
where $\S(\Z^2)$ are the sequences of rapid decay. The algebra $\A_\theta$ deforms $C^\infty(\bT^2)$ since for $\theta=0$ the above expansion agrees with the usual Fourier series expansion on $\bT^2$ in terms of the generating unitaries.

Both $A_\theta$ and $\A_\theta$ carry an action $\sigma$ of $\bT^2$ by automorphisms, given on the generators by
$$
\sigma_t(U_1) = e^{i t_1} U_1; \qquad \sigma_t(U_2) = e^{it_2} U_2 .
$$
for $t = (t_1,t_2) \in \bT^2$.

\medskip

Now, consider an arbitrary compact Riemannian spin manifold $M$ that carries a (smooth) action of a $2$-torus. We can then `insert' the structure of the noncommutative torus in $M$, in the following, precise sense. 


\begin{defn}
\label{defn:toric-def}
We define the $C^*$-algebra $C(M_\theta)$ as the following invariant functions from $M$ with values in $A_\theta$:
$$
C(M, A_\theta)^{\bT^2} := \left\{ f \in C(M, A_\theta): f(t \cdot x) = \sigma_t(f(x)) , \quad x \in M \right\}
$$
\end{defn}
Similarly, we define $C^\infty(M_\theta):= C^\infty(M, \A_\theta)^{\bT^2}$ as $\bT^2$-equivariant smooth functions from $M$ with values in $\A_\theta$.

However, not only can we deform the $*$-algebras $C(M)$ and $C^\infty(M)$, also the spinor bundle, Dirac operator and charge conjugation as they appear in Example \ref{ex:can-st} can be deformed.
\begin{thm}[Connes--Landi \cite{CL01}]
Let $M$ be a compact Riemannian spin manifold, and let $(C^\infty(M), \H=L^2(M,\S), D; J)$ be the corresponding canonical spectral triple. Then there is a representation of $C^\infty(M_\theta)$ on $\H$ such that 
$$
(C^\infty(M_\theta), \H,D;J)
$$
is a real spectral triple. 
\end{thm}
Note that $\H,D$ and $J$ are unchanged, it is only the $*$-algebra and its representation on $\H$ that are deformed. For this reason, these deformations are examples of {\em isospectral deformations}.

Let us then consider the gauge theory that corresponds to this real spectral triple. We distinguish two cases corresponding to $\theta$ being rational or irrational. These two cases require completely different techniques and yield entirely different results.

For $\theta$ rational we have the following result, due to \'Ca\'ci\'c in \cite[Theorem 4.28]{Cac14}. If $p,q$ are coprime and $\theta=p/q$, we set $\Gamma_\theta = \Z/q \Z$. 
\begin{thm}
\label{thm:cacic}
We have the following equivalence of spectral triples:
$$
(C^\infty(M_\theta),L^2(M,\S),D) \cong \Gamma^\infty(M/\Gamma_\theta,B),L^2(M/\Gamma_\theta, \pi_* \S \otimes B), \pi_*D)
$$
in terms of the projection map $\pi: M \to M/\Gamma_\theta$ and a $*$-algebra bundle $B :=M \times_{\Gamma_\theta} M_{q}(\C)$ with base space $M/\Gamma_\theta$, for a suitable action of $\Gamma_\theta$ on $M_{q}(\C)$ (see \cite{Cac14} for full details). 
\end{thm}

Let us confront this with our gauge theory interpretation in terms of the commutative subalgebra $C(M_\theta)_J$ in $C(M)$. 

\begin{prop}
\label{prop:center-theta}
For the real spectral triple $(C^\infty(M_\theta), \H, D; J)$ we have for $\theta$ rational 
$$
C(M_\theta)_J = Z(C(M_\theta)).
$$
Moreover, in this case $C(M_\theta)_J \cong C(M/\Gamma_\theta)$.
\end{prop} 
\proof 
First, $C(M_\theta)_J \subset Z(C(M_\theta))$ by Proposition \ref{prop:subalg}. The converse inclusion is obtained as follows. We have $Z(C(M_\theta)) \cong C(M/\Gamma_\theta)$ because $Z(M_{q}(\C)) = \C$ for all fibers. Moreover, $C(M/\Gamma_\theta) = C(M)^{\Gamma_\theta}$ is a subalgebra of $C(M)$, all of whose elements satisfy the commutation relation $aJ = Ja^*$ ({\it cf.} Example \ref{ex:AJ-can-st}).
\endproof

Hence, the bundle $B=M \times_{\Gamma_\theta} M_{q}(\C)  \to M/\Gamma_\theta$ is the sought-for $C^*$-bundle on which to define our gauge theory. Theorem \ref{thm:cacic} tells us that the $C^*$-algebra $C(M_\theta)$ is isomorphic to the space of continuous sections of $B$ ---in concordance with our Theorem \ref{thm:usc}--- and for the gauge group we actually have the following result:
$$
\G(C(M_\theta),\H;J) \cong \Gamma(M/\Gamma_\theta, M \times_{\Gamma_\theta} PU(q)),
$$ 
if $M/\Gamma_\theta$ is simply connected. In other words, we are considering a $PU(q)$-gauge theory as in Section \ref{subsect:ym}. This Lie group acts on the fiber $M_{q}(\C)$ of $B$ in the adjoint representation. The fact that the above spectral triple is thus an example of an almost-commutative spectral triple in the sense of \cite{BoeS10,Cac11,BD13} was already noticed in \cite{Cac14}.

\bigskip

Let us now proceed with the case that $\theta$ is irrational.

\begin{prop}
\label{prop:center-Mtheta}
For the real spectral triple $(C^\infty(M_\theta), \H, D; J)$ we have for $\theta$ irrational
$$
C(M_\theta)_J = Z(C(M_\theta)).
$$
Moreover, in this case $C(M_\theta)_J \cong C(M/\bT^2)$. 
\end{prop}
\proof
First note that $Z(C(M_\theta)) = C(M_\theta)^{\bT^2}$, essentially because the center of $A_\theta$ is trivial if $\theta$ is irrational ({\em cf.} \cite[Proposition 3]{CD02}). Moreover, since $C(M)^{\bT^2}$ is unchanged under the deformation, as well as $J$, we find that $C(M_\theta)^{\bT^2} \cong C(M)^{\bT^2}$ is contained in $C(M_\theta)_J$ which also proves the second statement.
\endproof

This allows us to conclude with Theorem \ref{thm:usc} that $C(M_\theta)$ is isomorphic to the $C^*$-algebra $\Gamma(M/\bT^2, \fB^{M_\theta})$ of continuous sections of an upper semi-continuous $C^*$-bundle $\fB^{M_\theta} \to M/\bT^2$ and that $\G(C(M_\theta),\H;J)$ acts by vertical automorphisms on $\fB^{M_\theta}$. This also follows from the more general results of \cite{BM12} showing that torus-covariant $C(X)$-algebras are deformed to torus-covariant $C(X)$-algebras. Here a torus-covariant algebra is a $C(X)$-algebra which carries an action of $\bT^2$ that commutes with $C(X)$. In particular, this applies to the $C(M/\bT^2)$-algebra $C(M)$, deforming to the $C(M/\bT^2)$-algebra $C(M_\theta)$. In fact, even more can be said in this case.

\begin{thm}
\label{thm:toric-bundle}
The above $C^*$-bundle $\fB^{M_\theta} \to M/\bT^2$ is a continuous $C^*$-bundle. Moreover, its fibers are given by the following $C^*$-algebras:
$$
\fB^{M_\theta}_x  \cong C(\bT^2/\bT^2_x,A_\theta)^{\bT^2},
$$
for $x \in M/\bT^2$ having isotropy group $\bT^2_x \subseteq \bT^2$.
\end{thm}
\proof
In addition to upper semi-continuity, in \cite[Proposition 5.1]{BM12} lower semi-continuity is shown to hold under some additional conditions. In fact, since the $\bT^2$-orbit space of $M$ is Hausdorff, Corollary 5.3 in {\it loc.~cit.~}implies that the Rieffel deformation $C(M_\theta)$ of $C(M)$ can be expressed as a continuous field of $C^*$-algebras over this orbit space. In other words, it is the $C^*$-algebra of sections of a continuous $C^*$-bundle over $M/\bT^2$. The second claim follows from \cite[Corollary 6.2]{BM12}.
\endproof

Hence, the spectral triple $(C(M_\theta),\H,D;J)$ yields a gauge theory defined in terms of a $C^*$-bundle $\fB^{M_\theta} \to M/\bT^2$. The gauge group $\G(C(M_\theta),\H;J)$ is parametrized by unitaries in $C(M_\theta)$ and acts vertically on the bundle $\fB^{M_\theta}$. We now determine the bundle structure of the gauge group, thereby making use of Theorem \ref{thm:lift} above.

\begin{prop}
\label{prop:gauge-theta}
There exists a subbundle $\tilde{\G\B^{M_\theta}} \subset \U\B^{M_\theta}$ that is a covering space of $\G\B^{M_\theta}$ for the quotient map $\U\B^{M_\theta}\to \G\B^{M_\theta}$. Consequently, if $M/\bT^2$ is simply connected we have
$$
\G(C(M_\theta),\H;J) \cong \Gamma(M/\bT^2,\G\B^{M_\theta}),
$$
where the fibers of $\G\B^{M_\theta}$ are given by 
$$
\G\B_x^{M_\theta} \cong \frac{\U(C(\bT^2/\bT^2_x,A_\theta)^{\bT^2})}{\U(\C)}; \qquad (x \in M/\bT^2).
$$
\end{prop}
\proof
From Theorem \ref{thm:toric-bundle} it follows that the fibers $\U\B_x^{M_\theta}$ of $\U\B^{M_\theta}$ are given by the topological groups $\U(C(\bT^2/\bT^2_x,A_\theta)^{\bT^2})$. We define a subbundle of $\U\B^{M_\theta}$ using the unique tracial state $\tau$ on $A_\theta$. First, consider the phase map $\varphi: A_\theta \to U(1)$ given by
$$
\varphi(a) = \frac{\tau(a)}{|\tau(a)|} ; \qquad (a \in A_\theta).
$$
It induces a phase map on the fibers of $\B^{M_\theta}$ by composition:
\begin{align*}
\tilde \varphi: C(\bT^2/\bT^2_x,A_\theta)^{\bT^2}) &\to U(1),\\
f &\mapsto \varphi \circ f.
\end{align*}
We then define a subbundle $\tilde{\G\B^{M_\theta}} \subset \U\B^{M_\theta}$ by giving its fibers:
$$
\tilde{\G\B^{M_\theta}}_x = \left\{ u \in \U : \tilde \varphi(u)=1 \right\}.
$$
For $\tilde{\G\B^{M_\theta}}$ to be a covering space of $\G\B^{M_\theta}$, we determine the kernel of the quotient map $\U\B^{M_\theta} \to \G\B^{M_\theta}$, intersected with $\tilde{\G\B^{M_\theta}}$. In fact, being in the kernel amounts to $u \in \U(\C)$ so that $\varphi(u)=1$ implies that then $u=1$. Hence, $\tilde{\G\B^{M_\theta}}$ is a one-fold covering of $\G\B^{M_\theta}$. 

If $M/\bT^2$ is simply connected, then  Theorem \ref{thm:lift} and \ref{thm:toric-bundle} combine to prove the second statement. 
\endproof

The above result allows for the following explicit bundle description of the group of inner automorphisms of $C(M_\theta)$. 
Note that $M/\bT^2$ is simply connected when $M$ is, {\it cf.} \cite[Corollary 6.3]{Bre72}. 
\begin{corl} 
\label{corl:inner-theta}
If $M/\bT^2$ is simply connected, then 
$$
\Inn(C(M_\theta)) \cong \Gamma(M/\bT^2 ,\G\fB^{M_\theta}).
$$
\end{corl}
\proof
In Proposition \ref{prop:center-theta} we have already established that $Z(C(M_\theta)) \cong C(M_\theta)_J$. Hence Corollary \ref{corl:gauge-ncg} applies and gives the group isomorphism $\Inn(C(M_\theta)) \cong \G(C(M_\theta),\H;J)$. Combining this with Proposition \ref{prop:gauge-theta} yields the desired result.
\endproof

\subsubsection{The toric noncommutative 3-sphere}
\label{subsubsect:S3}
Let us consider an explicit example of the above construction of toric noncommutative manifolds, namely a deformation of the 3-sphere, originally appearing in the $C^*$-context in \cite{Mat91a,Mat91}. Its Riemannian spin geometry was later explored in \cite{CL01}. See also \cite{CD02} for a more general family of 3-spheres. 


Consider the three-sphere $\bS^3 = \{ (a,b) \in \C^2: | a|^2 + | b|^2=1\}$ and parametrize by toroidal coordinates:
\begin{gather*}
a = e^{ i t_1} \cos \chi; \qquad 
b = e^{ i t_2} \sin \chi
\end{gather*}
where $0 \leq t_i \leq 2 \pi$ parametrize a 2-torus and $0 \leq \chi \leq \pi/2$. There is thus a natural action of $\bT^2$.

\begin{prop}
The $C^*$-algebra $C(\bS^3_\theta) :=C(\bS^3, A_\theta)^{\bT^2}$ (see Definition \ref{defn:toric-def}) is isomorphic to the $C^*$-algebra generated by $\alpha,\alpha^*$ and $\beta,\beta^*$ subject to the following conditions:
\begin{gather*}
\alpha \alpha^* = \alpha^* \alpha,\qquad \beta \beta^* = \beta^* \beta,\qquad
 \beta\alpha = e^{2\pi i \theta} \alpha \beta ,\qquad  \alpha \alpha^* + \beta \beta^* = 1.
\end{gather*}
\end{prop}
\proof
Note that $C(\bS^3,A_\theta)^{\bT^2}$ is generated by $a u_1$ and $b u_2$, with $a,b$ the generating functions of $C(\bS^3)$. One computes that the commutation relations are satisfied for $\alpha=a u_1$ and $\beta=bu_2$, {\it e.g.}
$$
(bu_2)(au_1) =  ab u_2 u_1 = e^{2\pi i\theta} ab u_1 u_2 = e^{2\pi i \theta} (au_1)(bu_2).
$$
\endproof

Suppose that $\theta$ is irrational. We determine the base space $C(\bS^3_\theta)_J \cong C(\bS^3)^{\bT^2}$. In terms of the toroidal parametrization of $\bS^3$ we readily find that the $\bT^2$-invariant subalgebra is given by functions depending on $|\alpha|$ and $|\beta|$. In other words,
$$
C(\bS^3_\theta)_J \cong C[0,\pi/2]
$$
corresponding to the angle $0\leq \chi \leq \pi/2$.



We also explicitly determine the fibers of the $C^*$-bundle $\fB^{\bS^3_\theta} \to [0,\pi/2]$ for which $C(\bS^3_\theta) \cong \Gamma([0,\pi/2],\fB^{\bS^3_\theta})$. At a point $\chi \in [0,\pi/2]$ we have fiber
$$
\fB^{\bS^3_\theta}_\chi \cong C^*\langle u_1 \cos \chi, u_2 \sin \chi \rangle,
$$
as the $C^*$-algebra generated by $u_1 \cos \chi$ and $u_2 \sin \chi$ in terms of unitaries $u_1, u_2$ that satisfy the defining relation $u_2 u_1 = e^{2\pi i \theta} u_1 u_2$. In other words, if $\chi \in (0,\pi/2)$, the $C^*$-algebra $\fB^{\bS^3_\theta}$ is isomorphic to the noncommutative torus $A_\theta$. However, at the endpoints we have that
\begin{align*}
\fB^{\bS^3_\theta}_{\chi=0}& \cong C^*\langle u_1\rangle \cong C(\bS^1),
\intertext{and}
\fB^{\bS^3_\theta}_{\chi=\pi/2}& \cong C^*\langle u_2\rangle \cong C(\bS^1).
\end{align*} 
This can also be seen from Theorem \ref{thm:toric-bundle}. Indeed, for $\chi \in (0,\pi/2)$ the torus action is free so that the fiber is $C(\bT^2,A_\theta)^{\bT^2} \cong A_\theta$. However, the endpoints $\chi=0,\pi/2$ have isotropy groups $\bT \subset \bT^2$ so that the fiber becomes $C(\bT,A_\theta)^{\bT^2} \cong C(\bS^1)$.

This bundle picture gains in perspective when looking at the inflation and deflation of the two circle directions in the toroidal parametrization of $\bS^3$. Namely, for the points in $\bS^3$ corresponding to $\chi \in (0,\pi/2)$ the deformation amounts to replacing the $\bT^2$-orbits by noncommutative tori. This gives as fibers over these points the algebra $A_\theta$.  
However, as $\chi$ approaches $0$ or $\pi/2$, one of the circle directions in the $\bT^2$-orbits becomes increasingly smaller, whereas the other becomes increasingly larger. Eventually, at the points $\chi=0,\pi/2$ one of the circles has shrunk to a point, giving the $C^*$-algebra $C(\bS^1)$ as the fiber over these two points (see Figure \ref{fig:bundle-S3theta}). 

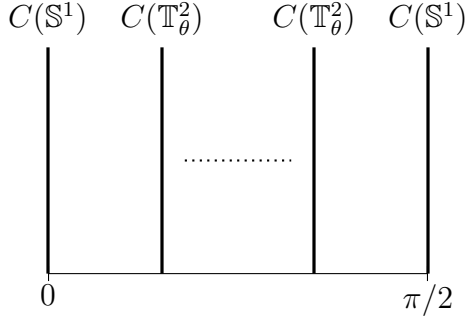
\begin{figure}
\begin{tikzpicture}
\draw[|-|] (0,0) node[anchor=north] {$0$} -- (5,0) node[anchor=north]{$\pi/2$};
\draw[very thick] (0,0) -- (0,3) node[anchor=south]{$C(\bS^1)$};
\draw[very thick] (1.5,0) -- (1.5,3) node[anchor=south]{$C(\bT^2_\theta)$};
\draw[very thick] (3.5,0) -- (3.5,3) node[anchor=south]{$C(\bT^2_\theta)$};
\draw[thick,dotted] (1.8,1.5)--(3.2,1.5);
\draw[very thick] (5,0) -- (5,3) node[anchor=south]{$C(\bS^1)$};
\end{tikzpicture}
\caption{The $C^*$-bundle on $[0,\pi/2]$ with section algebra isomorphic to $C(\bS^3_\theta)$.}
\label{fig:bundle-S3theta}
\end{figure}

Finally, since $[0,\pi/2]$ is simply connected, Proposition \ref{prop:gauge-theta} implies that the gauge group $\G(C(\bS^3_\theta), \H;J)$ is isomorphic to the space of continuous sections of a group bundle $\G\B^{\bS^3_\theta}$. Let us explicitly determine its fibers. At the end-points we have 
\begin{gather*}
\G\fB^{\bS^3_\theta}_{\chi=0} \cong \frac{\U(C(\bS^1))}{U(1)}, \qquad \G\fB^{\bS^3_\theta}_{\chi=\pi/2} \cong\frac{\U(C(\bS^1))}{U(1)} ,
\intertext{whereas for $\chi \in (0,\pi/2)$ we have}
\G\fB^{\bS^3_\theta}_{\chi} \cong \frac{\U(A_\theta)}{U(1)}.
\end{gather*}
Note that since $A_\theta$ is a simple $C^*$-algebra, it follows that the fibers $\G\fB^{\bS^3_\theta}_{\chi}$ for $\chi \in (0,\pi/2)$ are isomorphic to the group $\Inn(A_\theta)$ of inner automorphisms of $A_\theta$. If we combine this with Corollary \ref{corl:inner-theta} we finally arrive at an explicit group bundle description of the inner automorphisms of $C(\bS^3_\theta)$.

\begin{rem}
In \cite{BMS13} a factorization of the noncommutative 3-sphere $\bS^3_\theta$ was obtained in the context of unbounded KK-theory. Instead of the base space $[0,\pi/2]$ that we derived above from the real structure on a spectral triple, there we described a factorization over a base space $\bS^2$ with fibers $C(\bS^1)$. The structure of a noncommutative principal Hopf fibration $\bS^3_\theta \to \bS^2$ allowed for a metric factorization, that is to say, for a factorization of the (round) Dirac operator on $\bS^3_\theta$ in terms of the (round) Dirac operator on $\bS^2$ and a (vertical) $U(1)$-generator. Here, we go one step further as far as the topology is concerned, essentially considering also $\bS^2$ as a bundle over $[0,\pi/2]$. Since $[0,\pi/2]$ is contractible, its $K_1$-group is trivial so that one does not expect a factorization that is similar to the one appearing in \cite{BMS13}.
\end{rem}

\subsubsection{The toric noncommutative 4-sphere}

We end this section with the example of a noncommutative 4-sphere, whose differential geometric structure was first studied in \cite{CL01}. It formed a key example in the description of gauge theories on noncommutative manifolds ({\it cf.} \cite{LS04,LS06,BS10,BLS11}).

The noncommutative 4-sphere $\bS^4_\theta$ fits in the general class of toric noncommutative manifolds as defined above. Moreover, it is a suspension of $\bS^3_\theta$, paralleling the classical construction. 

Indeed, in addition to the complex coordinates $a$ and $b$ we parametrize $\bS^4$ by a real coordinate $x \in [-1,1]$:
\begin{gather*}
a = e^{ i t_1} \cos \chi \cos \psi; \qquad 
b = e^{ i t_2} \sin \chi \cos \psi; \qquad
x = \sin \psi
\end{gather*}
where $0 \leq t_i \leq 2 \pi$ parametrize a 2-torus as before, $0 \leq \chi \leq \pi/2$ and $-\pi/2\leq \psi \leq \pi/2$.

\begin{prop}
The $C^*$-algebra $C(\bS^4_\theta) :=C(\bS^4, A_\theta)^{\bT^2}$ (see Definition \ref{defn:toric-def}) is isomorphic to the $C^*$-algebra generated by $\alpha,\alpha^*,\beta,\beta^*$ and a central self-adjoint element $x=x^*$, subject to the following conditions:
\begin{gather*}
\alpha \alpha^* = \alpha^* \alpha,\qquad \beta \beta^* = \beta^* \beta,\qquad
 \beta\alpha = e^{2\pi i \theta} \alpha \beta ,\qquad  \alpha \alpha^* + \beta \beta^* +x^2= 1.
\end{gather*}
\end{prop}


This makes $C(\bS^4_\theta)$ a suspension of $C(\bS^3_\theta)$, so that the $C^*$-bundle description of $\bS^4_\theta$ follows directly from that of $\bS^3_\theta$. The algebra $C(\bS^4_\theta)_J$ is isomorphic to the $\bT^2$-invariant subalgebra, which in this case is the $C^*$-algebra generated by $|\alpha|$, $|\beta|$ and $x$. Hence, the base space is the region in $\R^3$ with coordinates $r, s$ (both positive, corresponding to $|\alpha|$ and $|\beta|$) and $x \in \R$. That is, the base space $X$ is given by 
\begin{equation}
\label{eq:base-S4theta}
X= \{ (r,s,x)  \in \R^3: r^2+s^2+x^2=1 , \quad r,s\geq 0 \}.
\end{equation}
\begin{figure}
\begin{tikzpicture}
\pgftext[center,at={\pgfpoint{0}{0}}]{\includegraphics[scale=.4]{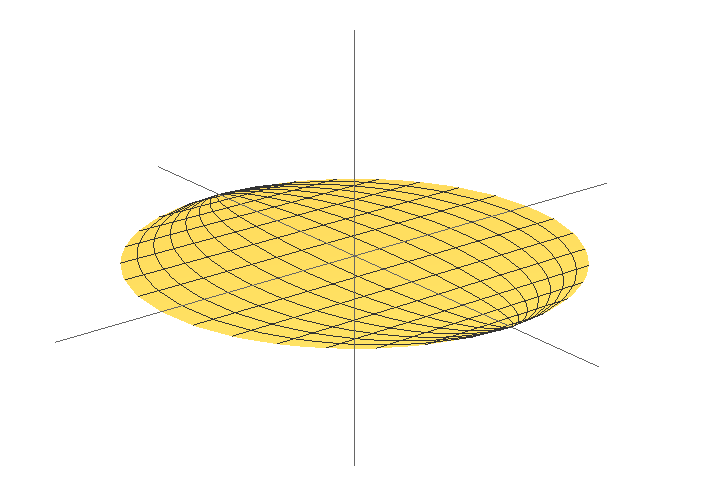}}
\node at (3.5,-1.3) {$x$};

\node at (-4,-1.5) {$r-s$};
\draw[very thick] (-2,.8) -- (-2,1.8) node[anchor=south]{$\C$};
\draw[very thick] (2.1,-1.05) -- (2.1,-0.05) node[anchor=south]{$\C$};
\draw[very thick] (-1.85,-1.2) -- (-1.85,-.2) node[anchor=south]{$C(\bS^1)$};
\draw[very thick] (-3.4,-.15) -- (-3.4,.85) node[anchor=south]{$C(\bS^1)$};
\draw[very thick] (2.7,.45) -- (2.7,1.45) node[anchor=south]{$C(\bS^1)$};
\draw[very thick] (.2,1) -- (.2,2) node[anchor=south]{$C(\bS^1)$};
\draw[very thick] (-.9,.3) -- (-.9,1.3) node[anchor=south]{$C(\bT^2_\theta)$};
\draw[very thick] (1.5,.2) -- (1.5,1.2)  node[anchor=south]{$C(\bT^2_\theta)$};
\draw[very thick] (.65,-1) -- (.65,0) node[anchor=south]{$C(\bT^2_\theta)$};
\end{tikzpicture}
\caption{The $C^*$-bundle on base space $X$ of Equation \eqref{eq:base-S4theta} with section algebra isomorphic to $C(\bS^4_\theta)$.}
\label{fig:bundle-S4theta}
\end{figure}
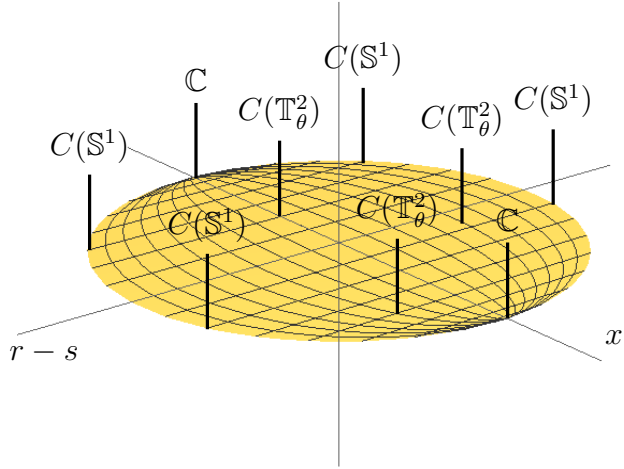
This is indeed a (smooth) suspension of the interval $[0,\pi/2]$. In fact, $X$ can be parametrized by 
$$
r = \cos\chi\cos\psi, \qquad s = \sin \chi \cos \psi, \qquad x= \sin \psi.
$$

For $|\psi| < \pi/2$ the fiber $\fB^{\bS^4_\theta}_{\chi,\psi}$ is isomorphic to the fiber $\fB^{\bS^4_\theta}_{\chi}$ that we derived for $\bS^3_\theta$ on the interval $[0,\pi/2]$. However, for $\psi = \pm \pi/2$ the entire fiber is reduced to the trivial $C^*$-algebra, {\it i.e.} $\fB^{\bS^4_\theta}_{\chi,\psi=\pm\pi/2} \cong \C$. The resulting $C^*$-bundle $\fB^{\bS^4_\theta} \to X$ satisfies 
$$
\Gamma(X,\fB^{\bS^4_\theta})\cong C(\bS^4_\theta)
$$
and is depicted schematically in Figure \ref{fig:bundle-S4theta}.

Since $X$ is simply connected we can apply Proposition \ref{prop:gauge-theta} to conclude that the gauge group $\G(C(\bS^4_\theta), \H;J)$ is isomorphic to $\Gamma(X,\G\B^{\bS^4_\theta})$. The fibers of the group bundle $\G\B^{\bS^4_\theta}$ are given by the trivial groups for $\psi=\pm \pi/2$, by $\U(C(\bS^1))/U(1)$ if $|\psi|< \pi/2$ and $\chi = 0$ or $\pi/2$, and by $\Inn(A_\theta)$ if $|\psi|< \pi/2$ and $\chi \in (0,\pi/2)$. Again, when combined with Corollary \ref{corl:inner-theta}, this gives an explicit group bundle description of the group of inner automorphisms of $C(\bS^4_\theta)$.

\section{Outlook}

Besides the examples discussed in this paper, there are many spectral triples in the literature for which our bundle picture could give a handle on the corresponding generalized gauge theory. Moreover, this could lead to an explicit group bundle description of the group of inner automorphisms of the pertinent $C^*$-algebra.

An interesting class of examples is given by real spectral triples on quantum groups or quantum homogeneous spaces, such as \cite{DS03,DLPS04,DLSSV04,SDLSV05,DALW06}. Another natural class of examples are (real spectral triples on) continuous trace $C^*$-algebras ({\it cf.} \cite{RW98} for a definition), especially in the case of non-vanishing Dixmier--Douady class. The latter condition would bring us beyond the continuous trace $C^*$-algebras that are Morita equivalent to $C(X)$, which (in the finite-rank case) reduces to the gauge theory described in Section \ref{subsect:ym}.

More generally, one could study the gauge theories coming from  {\em KK-fibrations} that were introduced in \cite{ENO09}: a $C^*$-bundle $\fB \to X$ is defined to be a $KK$-fibration if for any compact contractible space $\Delta$, any continuous map $f: \Delta \to X$ and any $z \in \Delta$, the evaluation map $\ev_z:\Gamma(\Delta,f^*\fB) \to \fB_{f(z)}$ is a KK-equivalence. In particular, this gives rise to a so-called {\em K-fibration}, in the sense that we have for the corresponding K-groups:
$$
K_\ast(\Gamma(\Delta,f^*\fB)) \cong K_\ast (\fB_{f(z)}).
$$
The $C^*$-bundles constructed in Theorem \ref{thm:toric-bundle} to describe toric noncommutative manifolds $C(M_\theta)$ do not fall in this class. Indeed, if we consider the $C^*$-bundle corresponding to $\bS^3_\theta$ as discussed in Section \ref{subsubsect:S3}, we can take $\Delta=[0,\pi/2]=X$, $f$ the identity map and $z=0$. Then $\Gamma(\Delta,\bS^3_\theta) \cong C(\bS^3_\theta)$ and $\fB_0 \cong C(\bS^1)$ and these two $C^*$-algebras have different $K^1$-groups.

\newcommand{\noopsort}[1]{}\def\cprime{$'$}

\end{document}